\documentclass[12pt]{article}

\usepackage{amsmath, amssymb, amsthm}
\usepackage{mathrsfs}
\usepackage{bm}

\usepackage[utf8]{inputenc}
\usepackage[T1]{fontenc}
\usepackage{lmodern}
\usepackage{microtype}
\usepackage{setspace}
\usepackage{hyperref}
\usepackage{todonotes}
\usepackage{epigraph}

\newtheorem{theorem}{Theorem}[section]
\newtheorem{lemma}[theorem]{Lemma}
\newtheorem{corollary}[theorem]{Corollary}

\title{\vspace{-2cm}On the impossibility of discovering \\ a formula for primes using AI}

\author{Alexander Kolpakov \\ \href{mailto:me@somewhere.com}{kolpakov.alexander@gmail.com} 
   \and Aidan Rocke \\ \href{mailto:me@somewhere.com}{rockeaidan@gmail.com} }

\date{\today}

\usepackage[a4paper, margin=2.5cm]{geometry}

\usepackage{charter}         

\begin{document}

\maketitle

\begin{abstract}
The present work explores the theoretical limits of Machine Learning (ML) within the framework of Kolmogorov's theory of Algorithmic Probability, which clarifies the notion of entropy as Expected Kolmogorov Complexity and formalizes other fundamental concepts such as Occam's razor via Levin's Universal Distribution. As a fundamental application, we develop Maximum Entropy methods that allow us to derive the Erd\H{o}s--Kac Law and Hardy--Ramanujan theorem in Probabilistic Number Theory, and establish the impossibility of discovering a formula for primes using Machine Learning via the Prime Coding Theorem.  
\end{abstract}

\newpage 

\epigraph{God made the integers; all else is the work of man.}{Leopold Kronecker}

\newpage 

\tableofcontents

\newpage 

\section{Compressing human concepts}

You meet a French mathematician at the Reykjavik airport with a million things on his mind but 
at any moment he is only thinking of one particular topic. Assuming that this list of concepts is known a priori, what is the minimum number of binary questions, asked in sequential order, that you would need to determine what he is thinking about? In the worst case, 

\begin{equation}
\log_2 (10^6) \approx 20  
\end{equation}

So we might as well play a game of $20$ questions. Moreover, the popularity of this game suggests that any human concept may be described using at most $20$ bits of information. If we may solve this particular inductive problem, might it be possible to solve the general problem of scientific induction?

\section{Kolmogorov's theory of Algorithmic Probability}

Using Kolmogorov's theory of Algorithmic Probability, we may apply Occam's razor to any problem of scientific induction including the 
sequential game of $20$ questions. However, it is easy to forget that this requires overcoming a seemingly insurmountable scientific obstacle which dates back to von Neumann: 

\begin{quote}
``My greatest concern was what to call it. I thought of calling it 'information,' but the word was overly used, so I decided to call it 'uncertainty.' When I discussed it with John von Neumann, he had a better idea. Von Neumann told me, 'You should call it entropy, for two reasons. In the first place your uncertainty function has been used in statistical mechanics under that name, so it already has a name. In the second place, and more important, no one really knows what entropy really is, so in a debate you will always have the advantage.'' -- Claude Shannon	
\end{quote}

This was accomplished through an ingenious combination of Shannon's Theory of Communication with Alan Turing's Theory of Computation. What emerged is the most powerful generalisation of Shannon's theory for algorithmics and coding theory, so Kolmogorov Complexity and Shannon Entropy share the same units, and Kolmogorov Complexity elucidates the Shannon Entropy of a random variable as its Expected Description Length. Furthermore, assuming that the Physical Church-Turing thesis is true, Kolmogorov's theory of Algorithmic Probability formalizes Occam's razor as it is applied in the natural sciences. 

In the case of our game, we may formulate Occam's razor using the four fundamental theorems of Algorithmic Probability before approximating Kolmogorov Complexity, which is limit--computable, 
using Huffman Coding in order to solve the game of $20$ questions. 

The implicit assumption here is that the second player is able to encode the knowledge of the first player because both players share a similar education and culture. 

\section{Fundamental Theorems of Algorithmic Probability}

These are known results \cite{li-vitanyi}, however we provide all proofs for completeness in the Appendix. 

\subsection{Kolmogorov's Invariance Theorem}

Let $U$ be a Turing--complete language that is used to simulate a universal Turing machine. Let $p$ be an input of $U$ that produces the binary string $x \in \{0,1\}^*$. Then the \textit{Kolmogorov Complexity} (or \textit{Minimal Description Length}) of $x$ is defined as

\begin{equation}
    K_U(x) = \min_{p} \{|p|: U\circ p = x\}
\end{equation}

where $U \circ p$ denotes the output of $U$ on input $p$.

Kolmogorov's Invariance Theorem states that the above definition is asymptotically invariant to the choice of $U$. Namely, any other Turing--complete language (or, equivalently, another universal Turing machine) $U'$ satisfies

\begin{equation}
\forall x \in \{0,1\}^* : \lvert K_U(x)-K_{U'}(x) \rvert \leq O(1)
\end{equation}

That is,

\begin{equation}
\forall x \in \{0,1\}^* : - c(U, U') \leq K_U(x) - K_{U'}(x) \leq c(U, U')
\end{equation}
for some positive constant $c(U, U')$ that depends only on $U$ and $U'$.  

\textbf{Interpretation:}

The minimal description $p$ such that $U \circ p = x$ serves as a natural representation of the string $x$ relative to the Turing--complete language $U$. While it may be demonstrated that $K_U(x)$ and therefore $p$, is not computable, the fourth fundamental theorem asserts that the Expected Kolmogorov Complexity is asymptotically equal to the Shannon entropy of a random variable. Hence, Kolmogorov Complexity is computable \textit{on average}.

\subsection{Levin's Universal Distribution} 

The \textit{Algorithmic Probability} of a binary string $x$ may be defined as the probability of $x$ being generated by $U$ on random input $p$, where $p$ is a binary string generated by fair coin flips:

\begin{equation}
    P(x) = \sum_{p \,:\, U\circ p = x} 2^{-|p|}
\end{equation}

However, this measure is not guaranteed to converge: we can choose one such input $p$ and use it as a prefix for some $p'$ that is about $\log_2 k$ bits longer than $p$ and such that $U$ produces the same binary string: $U\circ p' = x$. Then we find that: 

\begin{equation}
    P(x) \geq 2^{-|p|}\, \sum_k \frac{1}{k}
\end{equation}

for $k$ from any subset of integers. Thus, we can't guarantee that $P(x)$ converges.

Levin's idea effectively formalizes Occam's razor: we need to consider prefix--free Turing--complete languages only. Such languages are easy to imagine: if we agree that all documents end with an instruction that cannot appear anywhere else, then we have a prefix--free language. 

Given that any prefix--free code is uniquely-decodable, it satisfies the Kraft-McMillan inequality. Thus, we obtain Levin's Universal Distribution: 

\begin{equation}
2^{-K_U(x)} \leq P(x) = \sum_{p \,:\, U\circ p = x} 2^{-|p|} \leq 1
\end{equation}

where from hereon we consider $U$ to be prefix--free, and $K_U$ now corresponds to \textit{prefix--free Kolmogorov Complexity}. 

\subsection{Levin's Coding Theorem}

In the setting of prefix--free Kolmogorov complexity, Levin's Coding theorem states that 

\begin{equation}
-\log_2 P(x) = K_U(x) - O(1)
\end{equation}

Hence,

\begin{equation}
P(x) = \Theta\left( 2^{-K_U(x)} \right)
\end{equation}

\textbf{Interpretation:}

Relative to a prefix--free Turing--complete language $U$ (or, equivalently, a universal prefix--free Turing machine), the number of fair coin flips required to generate the shortest program that outputs $x$ is on the order of $\sim K_U(x)$. Thus, from a frequentist perspective, the entropy of the Universal 'a priori' Probability that we observe the event $x$ has a \textit{normal order} of:

\begin{equation}
 -\log_2 P(x) \sim K_U(x)
\end{equation}

Though we can't estimate Kolmogorov Complexity, we may use Lemma 7.2 to approximate the normal order of Kolmogorov Complexity. As an immediate consequence, we may estimate the normal order of Algorithmic Probability. 
Hence, we may reliably evaluate the Expected Kolmogorov Complexity of a random variable although Kolmogorov Complexity is not computable. 

It follows that Levin's Coding theorem allows us to formalize the notion of \textit{entropy of an event}. 

\newpage

\subsection{Maximum Entropy via Occam's razor}

Given a discrete random variable $X$ with computable probability distribution $P$, it holds that

\begin{equation}
\mathbb{E}[K_U(X)] = \sum_{x \in X} P(x) \cdot K_U(x) = H(X) + O(1)
\end{equation}

where $H(X)$ is the Shannon Entropy of $X$ in base $2$.

\textbf{Interpretation:}

The Shannon Entropy of a random variable in base $2$ is asymptotically equal to Expected Kolmogorov Complexity, which provides us with a precise answer to Von Neumann's original question. This theorem also allows us to assert that Kolmogorov Complexity is measurable on average where the expectation is calculated relative to Levin's Universal Distribution. 

Moreover, as an immediate consequence, we may deduce the Principle of Maximum Entropy. This follows first from the equivalence of Shannon's Source Coding theorem with the Asymptotic Equipartition Theorem, which we use to prove this fundamental theorem. Second, the Principle of Maximum Entropy that refers to constrained optimisation methods for estimating the source distribution from data is a specific application of Shannon's Source Coding theorem. Thus, machine learning systems that minimise the $KL$--Divergence are implicitly applying Occam's razor. 

But, what exactly do we mean by random variable? In a computable Universe the sample space of a random variable $X$ represents the state--space of a Turing Machine with unknown dynamics whose output sequence is computable. As the generated sequence is computable, it is finite--state incompressible in the worst-case i.e. a normal number. Hence, a random variable corresponds to a stochastic source that is finite--state random. 

This definition comes from a well--known correspondence between finite--state machines and normal numbers that establishes that a sequence is normal if and only if there is no finite--state machine that accepts it. 

\newpage 

\section{Preliminaries: the Game of 20 Questions}

\textbf{1.} The Game of $20$ Questions is played between Alice and Bob who are both assumed to be trustworthy and rational. Thus, Alice and Bob both perform sampling and inference using Levin's Universal Distribution $P$ over a shared alphabet $\mathcal{A}=\{a_i\}_{i=1}^n$. 

\textbf{2.} For the sake of convenience, we shall assume that $\mathcal{A}$ represents entries in the 
Britannica Encyclopedia and $\log_2 n \approx 20$. 

\textbf{3.} Bob selects an object $a \in \mathcal{A}$ and Alice determines the object by asking binary ``yes / no'' 
questions in a sequential manner, encoded using a prefix--free code, sampled from $P(\mathcal{A})$. 

\textbf{4.} Alice's goal is to minimize the expected number of questions which is equivalent to determine an $X \sim P(\mathcal{A})$ such that
\begin{equation}
\mathbb{E}[K_U(X)] \sim  H(X) \tag{*}
\end{equation}

\textbf{5.} In this setting, the Shannon Entropy may be understood as a measure of hidden information and we shall show that (*) has a solution using Huffman Coding. 

\subsection{An approximation to the Universal Distribution}

\textbf{Input:}

An alphabet $\mathcal{A} = \{a_i\}_{i=1}^n$ of symbols and a discrete probability distribution $P(\mathcal{A}) = \{p_i\}_{i=1}^n$ with $p_i$ equal to the frequency of $a_i$.

\textbf{Output:}

A prefix--free code $C(P)=\{c_i\}_{i=1}^n$ where $c_i \in \{0,1\}^*$ is the codeword for $a_i$. 

\textbf{Goal:}

Let the loss function be $\mathcal{L}(C(P))= \sum_{i=1}^n p_i \cdot \lvert c_i \rvert$, that is the weighted length of code $C$. We want to solve the minimization problem:

\begin{equation}
    \min_{C = C(P)}\,\mathcal{L}(C)
\end{equation}

If the Turing--complete language $U$ is prefix--free, we obtain:

\begin{equation}
H(X) + O(1) = \mathbb{E}[K_U(X)] = \sum_{i=1}^n p_i \cdot K_U(c_i) \leq \sum_{i=1}^n p_i \cdot \lvert c_i \rvert
\end{equation}
where $X \sim P(\mathcal{A})$. 

We may obtain an approximation to the Expected Kolmogorov Complexity, the lower-bound, via an entropy coding method such as Huffman Coding. This yields the desired prefix--free code $C(P)$. 

\subsection{Huffman Coding} 

\subsubsection{General description}

\textbf{1.} The technique works by creating a binary tree of nodes where leaf nodes represent the actual bytes in the input data. 

\textbf{2.} A node may be either a leaf node or an internal node. 

\textbf{3.} Initially, all nodes are leaf nodes each of which represents a symbol and its frequency. 

\textbf{4.} Internal nodes represent links to two child nodes and the sum of their frequencies. 

\textbf{5.} As a convention, bit '0' represents following the left child and bit '1' represents following the right child. 

\subsubsection{Simplest algorithm}

The simplest coding algorithm uses a priority queue where the node with the lowest probability is given the highest priority: 

\begin{enumerate}
 \item[\textbf{1.}] Create a leaf node for each symbol and add it to the priority queue. 
 
 \item[\textbf{2.}] While there is more than one node in the queue: 
 \begin{enumerate}
 	\item[\textbf{a.}] 	Remove the two nodes of highest priority (i.e. lowest probability) from the queue. 
	\item[\textbf{b.}] Create a new internal node with these two nodes as children and with probability equal to the sum of the two nodes' probabilities. 
	\item[\textbf{c.}] Add the new node to the queue. 
 \end{enumerate}

 \item[\textbf{3.}] The remaining node is the root node and the tree is complete. 
\end{enumerate}

Assuming that the nodes are already sorted, the time complexity of this algorithm is $O(n)$. 

\subsection{Discussion} 

Does the solution found via Huffman Coding agree with our intuitions? 

Assuming that internal nodes are given labels $v \in [1, 2 \cdot \lvert  \mathcal{A} \rvert]$ while leaf nodes are given labels $c_i \in \{0,1\}^*$ the information gained from any sequence of questions $S \subset [1, 2 \cdot \lvert  \mathcal{A} \rvert]$ may be determined from the entropy formula

\begin{equation}
H(S) = - \sum_{i \in S} p_i \cdot \log_2 p_i
\end{equation}

where the order of the internal nodes may be determined by sorting the vertices $i \in S$ with respect to their $\log$--probabilities $\{-\log_2 p_i \}_{i \in S}$. In principle, children of a parent node represent refinements of a particular concept, so the tree depth represents our depth of understanding. This degree of understanding may be measured in terms of the entropy $- p_i \cdot \log_2 p_i$. Hence, we have a satisfactory solution to the Game of $20$ questions. 

Zooming out, we may consider the ultimate impact of Kolmogorov's formalisation of scientific induction which Kolmogorov foretold \cite{kolmogorov1983combinatorial}: 

\begin{quote}
	Using his brain, as given by the Lord, a mathematician may not be interested in the combinatorial basis of his work. But the artificial intellect of machines must be created by man, and man has to plunge into the indispensable combinatorial mathematics. -- Kolmogorov (1983)
\end{quote}

In fact, Kolmogorov's theory of Algorithmic Probability may be viewed as a theory of machine epistemology. As for what may potentially limit the scope of machine epistemology relative to human epistemology, the \textit{big questions} section of \cite{hutter2000theory} may shed some light. 

\section{Acknowledgements}

We would like to thank Anders S\"odergren, Ioannis Kontoyiannis, Hector Zenil, Steve Brunton, Marcus Hutter, Cristian Calude, and Igor Rivin for constructive feedback in the preparation of this manuscript. 

\newpage

\section{Maximum Entropy methods for Probabilistic Number Theory}

\subsection{The Erd\H{o}s-Euclid theorem}

\textit{In essence, this proof demonstrates that the information content of finitely many primes is insufficient to generate all the integers. Originally due to Kontoyiannis \cite{kontoyiannis2008counting}.}

Let $\pi(N)$ be the number of primes that are less or equal to a given natural number $N$. Let us suppose that the set of primes $\mathbb{P}$ is finite so we have $\mathbb{P}=\{p_i\}_{i=1}^{\pi(N)}$ where $\pi(N)$ is constant for $N$ big enough. Then we can define a uniform integer--valued random variable $Z \sim U([1,N])$, such that 

\begin{equation}
Z = \big(\prod_{i=1}^{\pi(N)} p_i^{X_i}\big) \cdot Y^2 \tag{1}
\end{equation}

for some integer--valued random variables $1 \leq Y \leq \sqrt{N}$ and $X_i \in \{0,1\}$, such that $Z/Y^2$ is square--free. In particular, as we know that $Y \leq \sqrt{N}$, the upper bound for Shannon's Entropy from Jensen's inequality implies:

\begin{equation}
H(Y) \leq \log_2 \sqrt{N} = \frac{1}{2} \log_2 N \tag{2}
\end{equation}

Also, since $X_i$ is a binary variable, we have $H(X_i) \leq 1$. 

Using Kolmogorov's definition of Entropy, we obtain the asymptotic relation for the typical code length: 

\begin{equation}
\mathbb{E}[K_U(Z)] \sim H(Z) = \log_2 N	\tag{3}
\end{equation}

and we may deduce the following inequality: 

\begin{equation}
    H(Z) = H(Y, X_1, \ldots, X_{\pi(N)}) \leq H(Y) + \sum^{\pi(N)}_{i=1} H(X_i) \leq \frac{1}{2} \log_2 N + \pi(N) \tag{4}
\end{equation}

which implies:

\begin{equation}
\pi(N) \geq \frac{1}{2} \log_2 N \tag{5}
\end{equation}

This clearly contradicts the assumption that $\pi(N)$ is a constant for any natural $N$, and provides us with an effective lower bound on the prime counting function. 

\newpage

\subsection{Cheybshev's theorem via Algorithmic Probability}

\textit{An information-theoretic derivation of Chebyshev’s theorem (1852), an important precursor of the Prime Number Theorem, from the Maximum Entropy Principle. Another proof was given by Ioannis Kontoyiannis in \cite{kontoyiannis2008counting}.}

\textbf{Chebyshev's Theorem:}

We rediscover Chebyshev's theorem: 

\begin{equation}
H(X_{p_1},...,X_{p_{\pi(N)}}) = \sum_{p \leq N} \frac{1}{p} \cdot \log_2 p \sim \log_2 N \tag{1}
\end{equation}

which tells us that the expected information gained from observing a prime number in the interval $[1,N]$ is on the order of $\sim \log_2 N$. 

\textbf{Proof:}

For an integer sampled uniformly from the interval $[1,N]$ we may define its random prime factorization in terms of the random variables $X_p$: 

\begin{equation}
\forall Z \sim U([1,N]),\; Z = \prod_{p \leq N} p^{X_p} \tag{2}
\end{equation}

As we have no prior information about $Z$, it has the maximum entropy distribution among all possible distributions on $[1,N]$.  

While the  Kolmogorov Complexity of $Z$ is not computable, we may calculate its Expected Kolmogorov Complexity using Corollary~\ref{cor-integers} which tells us that almost all integers are incompressible: 

\begin{equation}
\mathbb{E}[K_U(Z)] \sim \mathbb{E}[\log_2 Z] = \frac{1}{N} \sum_{k=1}^N \log_2 k = \frac{\log_2(N!)}{N} \sim \log_2 N \tag{3}
\end{equation}

On the other hand, 

\begin{equation}
    \mathbb{E}[K_U(Z)] \sim \mathbb{E}[\log_2 Z] = \sum_{p \leq N} \mathbb{E}[X_p] \cdot \log_2 p \tag{4}
\end{equation}

By combining (3) and (4), we find:

\begin{equation}
\mathbb{E}[K_U(Z)] \sim \sum_{p \leq N} \mathbb{E}[X_p] \cdot \log_2 p \sim \log_2 N \tag{5}
\end{equation}

We also know that the geometric distribution maximizes the entropy of $X_p$ under the condition that $\mathbb{E}[X_p]$ is fixed \cite[Theorem 12.1.1]{cover-thomas}. Hence,

\begin{equation}
\mathbb{E}[X_p] = \sum_{k \geq 1} P(X_p \geq k) = \sum_{k \geq 1} \frac{1}{N} \left\lfloor \frac{N}{p^k} \right\rfloor \sim \frac{1}{p} \tag{6}
\end{equation}

\newpage

Thus, we rediscover Chebyshev's theorem: 

\begin{equation}
H\left(X_{p_1},...,X_{p_{\pi(N)}}\right) \sim \sum_{p \leq N} \frac{1}{p} \cdot \log_2 p \sim \log_2 N \tag{7}
\end{equation}

where this entropy formula is invariant to the choice of base of the logarithm.

\subsubsection{Corollary:}

As an important corollary, we may deduce that the base-e entropy of a typical prime number in the interval $[1,N]$ is on the order of $\sim \ln N$. We may begin by noting that the event that $Z$ is prime is given by the superposition of the event $X_p = 1$ of which there are $\pi(N)$ distinct possibilities and the null event $\sum_{p \leq N} X_p = 0$ which is unique. Hence, due to asymptotic independence we have: 

\begin{equation}
H(Z \in \mathbb{P}) = H(X_i = 1, \sum_{j \neq i} X_j = 0) \sim \sum_{i \leq N} H(X_i = 1) + H(\sum_{j \neq i} X_j = 0) 	\tag{8}
\end{equation}

where 

\begin{equation}
H(\sum_{j \neq i} X_j = 0) \sim -\sum_{p \leq N} \big(1-\frac{1}{p}\big) \cdot \ln \big(1-\frac{1}{p}\big) \sim \ln \ln N \tag{9}
\end{equation}

 and 
 
 \begin{equation}
 \sum_{p \leq N} H(X_p = 1) \sim \ln N \tag{10}
 \end{equation}
 
Therefore $H(Z \in \mathbb{P}) \sim \ln N$ which tells us that a typical prime number in the interval $[1,N]$ behaves as if any location in the interval $[1,N]$ is a priori equiprobable. 

\newpage 

\subsection{The Prime Coding Theorem}

\textit{In the following analysis, we consider an information-theoretic dual to the Prime Number Theorem via Chebyshev's theorem.}

We define the \textit{prime encoding} $\widehat{X_N} = \{\hat{x}_n\}_{n=1}^N$ where $\hat{x}_k = 1$ if $k \in \mathbb{P}$ and $\hat{x}_k = 0$ otherwise, as being the empirical realisation of a sequence of independent binary random variables  $X_N = \{x_n\}_{n=1}^N$. From an information-theoretic perspective, Chebyshev's theorem states that the average code length of a prime number in the interval $[1,N]$ is given by: 

\begin{equation}
H(X_1, ..., X_{p_\pi(N)}) \propto \log_2 N 	\tag{1}
\end{equation}

While Chebyshev's theorem is invariant to the choice of base of the logarithm, it tells us that the prime numbers are arranged uniformly in $[1,N]$ so from each $x_n \in X_N$ we may expect one bit of information:

\begin{equation}
	\frac{\mathbb{E}[K_U(X_N)]}{\pi(N)} \sim \frac{\log_2 2^N}{\pi(N)} \tag{2} 
\end{equation}

Furthermore, as each prime number contributes $H(X_{p_1},...,X_{p_{\pi(N)}})$ bits of information on average, in the limit of lossless compression Shannon's Source Coding theorem tells us that the expected code length for $X_N$ is given by: 

\begin{equation}
	\mathbb{E}[K_U(X_N)] \sim \pi(N) \cdot H(X_{p_1}, ..., X_{p_\pi(N)}) \tag{3}
\end{equation}

Finally, using the Prime Number Theorem $\frac{N}{\pi(N)} \sim \ln N$ we may deduce the Prime Coding theorem: 

\begin{equation}
\mathbb{E}[K_U(X_N)] \sim \pi(N) \cdot H(X_{p_1}, ..., X_{p_\pi(N)}) \sim N \tag{4}
\end{equation}

\begin{equation}
H(X_1, ..., X_{p_\pi(N)}) = \sum_{p \leq N} \frac{1}{p} \cdot \ln p \sim \ln N 	\tag{5}
\end{equation}

so the natural base is optimal for prime counting and we may conclude that the locations of all primes in $[1,N] \subset \mathbb{N}$ are statistically independent of each other. 

It follows that a machine learning model may not be reliably used to predict the location of the $N$--th prime number given prior knowledge of the location of the $N-1$ previous primes. 

\newpage 

\subsubsection{Corollary:}

In consequence, no prime formula may be approximated using Machine Learning. In particular, Riemann's Explicit formula for prime counting: 

\begin{equation}
\forall x \in \mathbb{N}, \pi (x)=\operatorname{R}(x)-\sum_{\rho}\operatorname{R}(x^{\rho}) \tag{6}
\end{equation}

\begin{equation}
\operatorname{R}(x)=\sum_{n=1}^\infty \frac{\mu (n)}{n}\operatorname{li}(x^{1/n}) \tag{7}
\end{equation}

is not learnable. Here $\mu (n)$ is the Möbius function, $\operatorname{li}(x)$ is the logarithmic integral function, and $\rho$ indexes every zero of the Riemann zeta function.

\newpage 

\subsection{The empirical density of the primes and their source distribution}

\textit{The Shannon Source Coding theorem informs us that the uniform source distribution is the correct generative model for all prime numbers.}

To clarify the Prime Coding Theorem in terms of random variables: 

\begin{equation}
H(X_N) \sim \pi(N) \cdot \ln N \sim N 	\tag{1}
\end{equation}

we must clarify the relation between the empirical density of primes and their source distribution of which they are a unique realisation. 

To be precise, we may model the prime encoding $\widehat{X_N}$ as the realisation of a non-stationary sequence of uniformly distributed random variables $X_N=\{x_n\}_{n=1}^N$. Thus, if the events $\widehat{x}_n \in \widehat{X_N}$ are arranged uniformly in the discrete interval $[1,N]$ then the expected information gained from observing all events in $\widehat{X_N}$ is given by: 

\begin{equation}
H(X_N) = \sum_{n=1}^N q_n \cdot H(X_{p_1},...,X_{p_{\pi(n)}}) \sim \sum_{n=1}^N q_n \cdot \ln n \sim N \tag{2}
\end{equation}

Furthermore, if we consider that each event is unique so $q_n$ is a unit fraction then we ought to model the sum of entropies: 

\begin{equation}
S_N = - \sum_{n=1}^N \ln q_n	 \tag{3}
\end{equation}

and to further our analysis we may define the Lagrangian function: 

\begin{equation}
\mathcal{L}(\lambda, q_n) = \sum_{n=1}^N q_n \cdot \ln n - \lambda \big(S_N + \sum_{n=1}^N \ln q_n \big)	\tag{4}
\end{equation}

Thus, we find that: 

\begin{equation}
\frac{\partial \mathcal{L}}{\partial q_n} = \ln n - \frac{\lambda}{q_n} = 0 \implies q_n = \frac{\lambda}{\ln n} \tag{5}	
\end{equation}

and in order to satisfy the asymptotic formula (1), we find that: 

\begin{equation}
\lambda \cdot N \sim N \implies q_n = \frac{1}{\ln n} \tag{6}	
\end{equation}

Therefore, the empirical density of primes which is described by the Prime Number Theorem:

\begin{equation}
	 \frac{\pi(N)}{N} \sim \frac{1}{\ln N} \tag{7}
\end{equation}
 
 is the natural consequence of a uniform source distribution. Moreover, as this is the simplest correct model for generating the prime numbers, Occam's razor also informs us that it is the correct generative model for all prime numbers. 
 
\newpage 

\subsection{Information-theoretic derivation of the Prime Number Theorem}

\textit{An information-theoretic derivation of the Prime Number Theorem via Occam's Razor.}

If we know nothing about the distribution of primes, in the worst case we may 
assume that each prime less than or equal to $N$ is drawn uniformly from $[1,N]$. So our source of primes is: 

\begin{equation}
X \sim U([1,N]) \tag{1}	
\end{equation}

where $H(X) = \ln N$ is the Shannon entropy of the uniform distribution. 

Now, we may define the prime encoding of $[1,N]$ as the binary sequence $\widehat{X}_N = \{\widehat{x}_n \}_{n=1}^N$ where $\widehat{x}_n = 1$ if $n$ is prime and $\widehat{x}_n = 0$ otherwise. With no prior knowledge, given that each integer is either prime or not prime, we have $2^N$ possible prime encodings in $[1,N] \subset \mathbb{N}$. As almost all binary strings of length $N$ are incompressible, the normal order of $K_U(\widehat{X}_N)$ must satisfy: 

\begin{equation}
K_U(\widehat{X}_N) \sim N \tag{2}
\end{equation}

Moreover, if there are $\pi(N)$ primes less than or equal to $N$ then the average number of bits per arrangement gives us the average amount of information gained from correctly identifying each prime in $[1,N]$ as: 

\begin{equation}
S_c = \frac{\log_2(2^N)}{\pi(N)} = \frac{N}{\pi(N)}  \sim \sum_{k = 1}^{N-1} \frac{1}{k} \lvert (k,k+1] \rvert \sim \ln N \tag{3}	
\end{equation}

as there are $k$ distinct ways to sample uniformly from $[1,k]$ and a frequency of $\frac{1}{k}$ associated with the event that $k \in \mathbb{P}$. 

In light of the last three arguments and Shannon's Noiseless Coding Theorem (a.k.a. Shannon's Source Coding Theorem), we may deduce that $K_U(\widehat{X}_N)$ must also satisfy the asymptotic relation:

\begin{equation}
K_U(\widehat{X}_N) \sim \pi(N) \cdot \ln N \sim N \tag{4}
\end{equation}

as the locations of the prime numbers are necessary and sufficient to encode $\widehat{X}_N$. 

Rearranging the last asymptotic relation, we rediscover the Prime Number Theorem: 

\begin{equation}
\frac{\pi(N)}{N} \sim \frac{1}{\ln N} \tag{5}
\end{equation}

Furthermore, this derivation of the empirical density which we may observe (5) via the source distribution which is not directly observable (1) indicates that the prime numbers are empirically distributed as if they were arranged uniformly. 

\newpage

\subsubsection{Corollary:}
 
 This generative model implies that independently of the amount of training data and computational resources at their disposal, if the best machine learning model predicts the next $N$ primes to be at $\{\hat{p}_i\}_{i=1}^N \in \mathbb{N}$ then for large $N$ this model's statistical performance will converge to a true positive rate that is no better than: 

\begin{equation}
\frac{1}{N}\sum_{i=1}^N \frac{1}{\hat{p}_i} \leq \frac{\ln N}{N} \tag{*}
\end{equation}

Hence, the true positive rate for any machine learning model converges to zero.

\newpage 

\subsection{An information-theoretic derivation of the Erd\H{o}s--Kac theorem}

\textbf{The Algorithmic Probability of a Prime Factor}

Given the integer $X \sim U([1,N])$ with random prime factorisation: 

\begin{equation}
X = \prod_{p \leq N} p^{X_p}	 \tag{1}
\end{equation}

we may define the Algorithmic Probability of the event $X_p \geq 1$ using the Prime Coding Theorem: 

\begin{equation}
\forall N \sim U([1,n]), \mathbb{E}[K_U(X_N)] \sim \pi(N) \cdot H(X_{p_1},...,X_{p_{\pi(N)}}) \sim N \tag{2}	
\end{equation}

\begin{equation}
H(X_{p_1},...,X_{p_{\pi(N)}}) = \sum_{i = 1}^{\pi(N)} H(X_{p_i}) \sim \sum_{p \leq N} \frac{1}{p} \cdot \ln p \sim \ln N \tag{3}
\end{equation}

From (2) and (3), we derive the Algorithmic Probability of the event $X_p \geq 1$ for large primes $p \in \mathbb{P}$: 

\begin{equation}
P(X_p \geq 1) \sim \frac{1}{p}	\tag{4}
\end{equation}

Likewise, the event that any two distinct primes $p,q \in \mathbb{P}$ are simultaneously observed occurs with Algorithmic Probability: 

\begin{equation}
P(X_p \geq 1 \land X_q \geq 1) \sim \frac{1}{p} \cdot \frac{1}{q}	\tag{5}
\end{equation}

as formulas (2) and (3) tell us that any two prime numbers $p,q \in \mathbb{P}$ are statistically independent. 

\textbf{The Expected number of Unique Prime Divisors}

For any integer $X \sim U([1,N])$, we may define its number of Unique Prime Divisors $w(X) = \sum_{p \leq N} X_p$ where $X_p = 1$ if $X \bmod p = 0$
and $X_p = 0$ otherwise. Thus, we may calculate the Expectation: 

\begin{equation}
\forall X \sim U([1,N]), \mathbb{E}[w(X)] = \sum_{p \leq N} 1 \cdot P(X_p \geq 1) + 0\cdot \big(1-P(X_p \geq 1)\big) \sim \sum_{p \leq N} \frac{1}{p} \sim \ln \ln N	\tag{6}
\end{equation}

where we used Mertens' Second theorem $\sum_{p \leq N} \frac{1}{p} \sim \ln \ln N$. 

\textbf{The Standard Deviation of $w(X)$}

As the random variables $X_p$ are independent, the variance of $w(X)$ is linear in $X_p$: 

\begin{equation}
\forall X \sim U([1,N]), \textrm{Var}[w(X)]	= \sum_{p \leq N} \mathbb{E}[X_p^2] - \mathbb{E}[X_p]^2 \sim \sum_{p \leq N} \big(\frac{1}{p} - \frac{1}{p^2} \big) \sim \ln \ln N \tag{7}
\end{equation}

since $\sum_{p \leq N} \frac{1}{p^2} \leq \frac{\pi^2}{6}$.

\newpage 

\textbf{The Erd\H{o}s--Kac theorem}

In order to prove the Erd\H{o}s--Kac theorem, it remains to show that $\omega(X) = \sum_{p \leq N} X_p$ satisfies the Lindeberg condition for the Central Limit Theorem:

\begin{equation}
\Lambda_N(\epsilon) = \sum_{p \leq N} \Big\langle \Big(\frac{X_p}{\sqrt{\text{Var}[\omega(X)]}}\Big)^2: \Big\lvert \frac{X_p}{\sqrt{\text{Var}[\omega(X)]}} \Big\rvert \geq \epsilon \Big\rangle \tag{8}
\end{equation}

\begin{equation}
\forall \epsilon > 0, \lim_{N \to \infty} \Lambda_N(\epsilon) = 0 \tag{9} 
\end{equation}

where $\langle \alpha:\beta \rangle$ denotes the expectation value of $\alpha$ restricted to outcomes $\beta$.

\textbf{Proof:}

Given our analysis of the Algorithmic Probability of a prime factor: 

\begin{equation}
P(X_p \geq 1) \sim \frac{1}{p} \tag{10}	
\end{equation}

where without normalising we may observe that: 

\begin{equation}
\sum_{p \leq N} P(X_p \geq 1 \lor X_p = 0) = \pi(N) \tag{11}	
\end{equation}

as we are a priori guaranteed $\pi(N)$ bits of information from determining the distinct prime factors of $X$. Normalising, so the frequency distribution is dimensionless: 

\begin{equation}
 f(X_p \geq 1) \sim \frac{1}{\pi(N)} \cdot \frac{1}{p} \tag{12}
\end{equation}

Given this normalised distribution, we may evaluate the expectation: 

\begin{equation}
\Sigma_N := \sum_{p \leq N} \Big\langle \Big(\frac{X_p}{\sqrt{\text{Var}[\omega(X)]}}\Big)^2: \Big\lvert \frac{X_p}{\sqrt{\text{Var}[\omega(X)]}} \Big\rvert \geq 0 \Big\rangle \tag{13}
\end{equation}

which is greater than or equal in value to $\Lambda_N(\epsilon)$ for any $\epsilon > 0$ since $\lim_{\epsilon \rightarrow 0} \Lambda_N(\epsilon) = \Sigma_N$. It follows that for large $N$ this expression simplifies to: 

\begin{equation}
\Sigma_N \sim \frac{1}{\ln \ln N} \sum_{p \leq N} f(X_p \geq 1) \sim 	\frac{1}{\pi(N) \cdot \ln \ln N} \sum_{p \leq N} \frac{1}{p} \sim \frac{1}{\pi(N)} \tag{14}
\end{equation}

and therefore the Lindeberg criterion is satisfied for any $\epsilon > 0$: 

\begin{equation}
0 \leq \lim_{N \to \infty} \Lambda_N(\epsilon) \leq \lim_{N \to \infty} \Sigma_N \sim \lim_{N \to \infty} \frac{1}{\pi(N)} = 0 \tag{15}
\end{equation}

Thus, we may conclude that: 

\begin{equation}
\forall X \sim U([1,N]), \frac{\omega(X)- \ln \ln N}{\sqrt{\ln \ln N}} \tag{16} 
\end{equation}

converges to the standard normal distribution $\mathcal{N}(0,1)$ as $N \rightarrow \infty$.

\newpage

\textbf{Discussion:}

This theorem is of great interest to the broader mathematical community as it 
is impossible to guess from empirical observations. In fact, it is far from certain that Erd\H{o}s and Kac would have proved the Erd\H{o}s--Kac theorem if its precursor, the Hardy--Ramanujan theorem, was not first discovered. 

More generally, at a time of Big Data and the imminent supremacy of AI, this 
theorem forces the issue of determining how some scientists were able to formulate correct theories based on zero empirical evidence. While the 
Erd\H{o}s-Kac theorem has the form of a statistical observation, the normal order 
of $w(X)$ only begins to emerge for $X \sim U([1,N])$ where $N \geq 10^{100}$. 

Within the current scientific paradigm, non-trivial scientific discoveries of this kind that are provably beyond the scope of scientific induction(and hence machine learning) do not yet have an adequate explanation. 

\newpage

\subsection{The Hardy--Ramanujan theorem}

The Hardy--Ramanujan theorem states that, given any $\epsilon > 0$, almost all integers satisfy: 

\begin{equation}
\forall n \in \mathbb{N}, |\omega(n)-\ln \ln n | < \epsilon \cdot \ln \ln n \tag{*}	
\end{equation}

It follows that $\omega(n)$, which measures the number of distinct prime factors of $n$, has normal order $\ln \ln n$. 

\textbf{Proof:}

Given the random variable $X \sim U([1,N])$ with particular realisation $\widehat{X} \in [1,N]$, we may encode the distinct prime factors of $\widehat{X}$ using the prime encoding:

\begin{equation}
\phi(\widehat{X}) := \{\hat{x_i}\}_{i=1}^{\pi(N)} \in \{0,1\}^{\pi(N)} \tag{1}
\end{equation}

On the other hand, $\omega(\widehat{X})$ measures the information gained from observing the distinct prime factors of $\widehat{X}$ since: 

\begin{equation}
\omega(\widehat{X}) = \sum_{i=1}^{\pi(N)} \hat{x_i} \tag{2}	
\end{equation}

This means that exactly $\omega(\widehat{X})$ binary questions asked in sequential order are necessary and sufficient to identify $\phi(\widehat{X})$. Hence, using Kolmogorov's Invariance theorem: 

\begin{equation}
K_U(\phi(\widehat{X})) \sim \omega(\widehat{X}) \tag{3}	
\end{equation}

As each prime factor of $\widehat{X}$ contributes exactly one bit of information, the Expected Kolmogorov Complexity of $\phi(X)$ is asymptotically: 

\begin{equation}
\mathbb{E}[K_U(\phi(X))] = \sum_{p \leq N} P(X_p \geq 1) \sim \sum_{p \leq N} \frac{1}{p} \sim \ln \ln N \tag{4}	
\end{equation}

Moreover, we may note that if $P(\phi(\widehat{X_i}))$ denotes the Algorithmic Probability of the event $\phi(X_i = \widehat{X_i})$ where $X_i \sim U([1,N])$, 

\begin{equation}
\frac{1}{n} \sum_{i=1}^n -\log_2 P(\phi(X_i = \widehat{X_i})) \rightarrow \mathbb{E}[K_U(\phi(X))] + \mathcal{O}(1) \tag{5}	
\end{equation}

where (5) follows from Levin's Coding theorem: 

\begin{equation}
-\log_2 P(\phi(\widehat{X_i})) \sim K_U(\phi(\widehat{X_i})) \tag{6}
\end{equation}

Furthermore, given the asymptotic relation: 

\begin{equation}
	\mathbb{E}[K_U(\phi(X))] \sim H(\phi(X)) \sim \ln \ln N \tag{7}
\end{equation}

we may apply the Asymptotic Equipartition Theorem. 

\newpage 

In particular, we may observe that for large $N \in \mathbb{N}$ each $\phi(X_i)$ is sampled i.i.d. from the finite set of possible prime encodings $\Omega = \{0,1\}^{\pi(N)}$ and the typical set satisfies: 

\begin{equation}
P(\{\phi(\widehat{X_i})\}_{i=1}^n \in A_{\epsilon}^n) \geq 1- \epsilon	\tag{8}
\end{equation}

\begin{equation}
	\lvert A_{\epsilon}^n \rvert \sim 2^{n \cdot \ln \ln N} \ll 2^{n \cdot \pi(N)} \tag{9}
\end{equation}

As an immediate consequence, we may apply the Asymptotic Equipartition Property(AEP) to a typical prime encoding $\{\hat{x_i}\}_{i=1}^{\pi(N)} \in B_{\epsilon}^{\pi(N)}$ which yields: 

\begin{equation}
P(\{\hat{x_i}\}_{i=1}^{\pi(N)} \in B_{\epsilon}^{\pi(N)}) \geq 1 - \epsilon \tag{10}
\end{equation}

\begin{equation}
	\lvert A_{\epsilon}^n \rvert \sim \lvert B_{\epsilon}^{\pi(N)} \rvert^n \implies \lvert B_{\epsilon}^{\pi(N)} \rvert \sim 2^{\ln \ln N} \tag{11}
\end{equation}

where $\epsilon \to 0$ as $\min(n,N) \to \infty$ due to the Law of Large Numbers. 

Finally, the AEP informs us that elements of the typical set $B_{\epsilon}^{\pi(N)}$ have the same asymptotic probability. Hence, the asymptotic relation:

\begin{equation}
-\log_2 P(\phi(\widehat{X_i})) \sim K_U(\phi(\widehat{X_i})) \sim \mathbb{E}[K_U(\phi(X))] \tag{12}
\end{equation}

holds almost surely and we may conclude that as $N \to \infty$:

\begin{equation}
\forall X \sim U([1,N]), P(\omega(X) \sim \ln \ln N) = 1 \tag{13}	
\end{equation}

\newpage

\section{Proofs of Fundamental Lemmas for Kolmogorov Complexity}

\textit{Although we don't prove these results for prefix-free languages, by 
Kolmogorov's Invariance theorem they readily generalise to prefix-free Universal 
Turing Machines.}

\begin{lemma}
    There exist algorithmically random strings.
\end{lemma}

\begin{proof}
    Let us suppose that for all $x, K_U(x) < n$. Then for all $x$ there exists some $p_x$ such that $U \circ p_x = x$ and $\lvert p_x \rvert < n$. Clearly, if $x \neq y$ then $p_x \neq p_y$. 
    
    Observe that there are $2^n - 1$ programs of length less than $n$, while the number of length $n$ strings is $2^n$. By the pigeonhole principle, if all strings of length $n$ have a program $p$ such that $|p| < n$ then there must be a program that produces two different strings. Thus, we have a contradiction. 
\end{proof}

\begin{lemma}\label{lemma-seq}
    Almost all finite strings are incompressible. 
\end{lemma}

\begin{proof}
    Let $c$ be some non--negative constant. The number of programs of length less than $n-c$ is $2^{n-c} - 1 < 2^{n-c}$, 
    which leaves us with $2^n-2^{n-c} = 2^n(1-2^{-c})$ programs of length greater than or equal to $n-c$. 
\end{proof}

\begin{corollary}\label{cor-integers}
    Almost all integers are algorithmically random. 
\end{corollary}

\begin{proof}
    The set of finite strings is countable so there is a bijective map from $\{0,1\}^*$ to $\mathbb{N}$ and we may define the Kolmogorov Complexity as a map from integers to integers, $K_U: \mathbb{N} \rightarrow \mathbb{N}$. As almost all finite strings are incompressible, it follows that almost all integers are algorithmically random: for almost all integers $n \in \mathbb{N}, K_U(n) \sim \log_2 n$.
\end{proof}

\begin{lemma}
    Kolmogorov Complexity is not computable.
\end{lemma}

\begin{proof}
  Let us suppose there exists a program $q$ of length $|q| = c$ such that for any binary string $x \in \{0, 1\}^*$ we have $q(x) = K_U(x)$ for some Turing--complete language $U$. Observe that there exist algorithmically random strings $x$ of length $|x| = n$, and that these strings can be arranged in the lexicographic order. Let this set be called $X$. This implies that we are able to use $q$ to determine whether $x \in X$ is the first string of length $n$ that satisfies
  \begin{equation}
    K_U(x) \geq n \tag{*}
  \end{equation}
  for any given $n$. As the set $X$ enumerable due to the computability of $q$ and the lexicographic order of inputs to $q$, the $n$--th string $x\in X$ that satisfies (*) requires $\sim \log_2 n$ bits of information so $K_U(x) \sim \log_2 n + c$. Hence, for sufficiently large $n$ we have a contradiction.  
\end{proof}

\newpage 

\section{Proofs of Fundamental Theorems for Algorithmic Probability}

\subsection{Proof of Kolmogorov's Invariance theorem:}

\textit{The following is taken from [5].}

From the theory of compilers, it is known that for any two Turing-Complete languages $U_1$ and $U_2$, there exists a compiler $\Lambda_1$ expressed in 
$U_1$ that translates programs expressed in $U_2$ into functionally-equivalent programs expressed in $U_1$. 

It follows that if we let $p$ be the shortest program that prints a given string $x$ then: 

\begin{equation}
K_{U_1}(x) \leq |\Lambda_1| + |p| \leq K_{U_2}(x) + \mathcal{O}(1)	\tag{1}
\end{equation}

where $|\Lambda_1| = \mathcal{O}(1)$, and by symmetry we obtain the opposite inequality. 

\subsection{Proof of Levin's Universal Distribution:}

This is an immediate consequence of the Kraft-McMillan inequality. 

Kraft's inequality states that given a sequence of strings $\{x_i\}_{i=1}^n$ there exists a prefix code with codewords $\{\sigma_i\}_{i=1}^n$ where $\forall i, |\sigma_i|=k_i$ if and only if: 

\begin{equation}
\sum_{i=1}^n s^{-k_i} \leq 1	 \tag{1}
\end{equation}

where $s$ is the size of the alphabet $S$. 

Without loss of generality, let's suppose we may order the $k_i$ such that: 

\begin{equation}
k_1 \leq k_2 \leq ... \leq k_n	\tag{2}
\end{equation}

Now, there exists a prefix code if and only if at each step $j$ there is at least one codeword to choose that does not contain any of the previous $j-1$ codewords as a prefix. Due to the existence of a codeword at a previous step $i<j, s^{k_j-k_i}$ codewords are forbidden as they contain $\sigma_i$ as a prefix. It follows that in general a prefix code exists if and only if: 

\begin{equation}
\forall j \geq 2, s^{k_j} > \sum_{i=1}^{j-1} s^{k_j - k_i}	\tag{3}
\end{equation}

Dividing both sides by $s^{k_j}$, we find: 

\begin{equation}
\sum_{i=1}^n s^{-k_i} \leq 1	 \tag{4}
\end{equation}

\newpage 

\subsection{Expected Kolmogorov Complexity equals Shannon Entropy}

Let's suppose that i.i.d. data $x_{1:n} \in \mathcal{A}^{(n)}$ are generated by sampling 
from the distribution $P_X$ on the finite set $\mathcal{A}$. Then it may be demonstrated 
that Expected Kolmogorov Complexity equals Shannon Entropy up to an additive constant: 

\begin{equation}
\sum_{x_{1:n} \in \mathcal{A}^(n)} P(X^{(n)}=x_{1:n}) \cdot K_U(X^{(n)} = x_{1:n}) = H(X) + \mathcal{O}(1) \tag{1}
\end{equation}

Now, by definition: 

\begin{equation}
H(X^{(n)}) = -\sum_{x_{1:n} \in \mathcal{A}^n} P(X^{(n)}=x_{1:n}) \cdot \log_2 P(X^{(n)}=x_{1:n}) \tag{2}
\end{equation}

where $P$ is a Universal Distribution that holds for our particular Observable Universe so that: 
\begin{equation}
-\log_2 P(X^{(n)}=x_{1:n}) = K_U(X^{(n)}=x_{1:n}) -\mathcal{O}(1) \tag{3}
\end{equation}

and Unitarity is guaranteed through an Oracle that identifies it with the Universal Wave Function, as $U$ is a computer that simulates the Observable Universe. 

\textbf{Proof:} 

If we carefully consider the Asymptotic Equipartition Theorem, 

\begin{equation}
\lim_{n \to \infty} P(x_{1:n} \in \mathcal{A_{\epsilon}}^{(n)}) = P\big(\lvert \frac{1}{n} \sum_{i=1} \log_2 \frac{1}{P(X_i^{(n)})}-H(X^{(n)}) \rvert \leq \epsilon \big) \rightarrow 1 \tag{4}
\end{equation}

we may define a natural measure on the atypical set $E_n = \mathcal{A} \setminus \mathcal{A_{\epsilon}}^{(n)}$: 

\begin{equation}
\mu_n = \frac{1}{\lvert E_n \rvert} \sum_{x_{1:n} \in E_n} P(X^{(n)}=x_{1:n}) \tag{5}
\end{equation}

Thus, we may observe that $\lim_{n \to \infty} \mu_n = 0$ so there must exist a positive exponent 
$\alpha > 1$ such that as $n \to \infty$: 

\begin{equation}
\mu_n^{-1} = \mathcal{O}\big(\lvert E_n \rvert^{\alpha} \big) \tag{6}
\end{equation}

Moreover, given that $K_U(X^{(n)}=x_{1:n}) \leq \log_2 |E_n|$ for $x_{1:n} \in E_n$: 

\begin{equation}
\sum_{x_{1:n} \in E_n} P(X^{(n)}=x_{1:n}) \cdot K_U(X^{(n)}=x_{1:n}) \leq \mu_n \cdot |E_n| \cdot \log_2 |E_n| = \mathcal{O}\big(\lvert E_n \rvert^{-\alpha} \big)  \tag{7}
\end{equation}

Hence, we may deduce: 

\begin{equation}
\forall x_{1:n} \in \mathcal{A_{\epsilon}}^{(n)}, K_U(X^{(n)}=x_{1:n}) \sim \log_2 |\mathcal{A_{\epsilon}}^{(n)}| \implies \mathbb{E}[K_U(X^{(n)})] \sim \log_2 |\mathcal{A_{\epsilon}}^{(n)}| \tag{8}
\end{equation}

Finally, due to the Asymptotic Equipartition Theorem we may determine that: 

\begin{equation}
H(X^{(n)}) \sim \mathbb{E}[K_U(X^{(n)})] \sim \log_2 |\mathcal{A_{\epsilon}}^{(n)}| \tag{9}
\end{equation}

\subsection{Corollary: Levin's Coding theorem}

If we carefully consider propositions (8) and (9), we may deduce: 

\begin{equation}
\forall x_{1:n} \in \mathcal{A_{\epsilon}}^{(n)}, -\log_2 P(X^{(n)}=x_{1:n}) \sim \log_2 |\mathcal{A_{\epsilon}}^{(n)}| \tag{10}
\end{equation}

\begin{equation}
\forall x_{1:n} \in \mathcal{A_{\epsilon}}^{(n)}, K_U(X^{(n)}=x_{1:n}) \sim \log_2 |\mathcal{A_{\epsilon}}^{(n)}| \tag{11}
\end{equation}

Thus, we may derive Levin's Coding theorem for i.i.d. time-series data: 

\begin{equation}
-\log_2 P(X^{(n)}=x_{1:n}) \sim K_U(X^{(n)}=x_{1:n}) \sim \log_2 |\mathcal{A_{\epsilon}}^{(n)}| \tag{12}
\end{equation}

which may be readily generalised to non-stationary data via the Asymptotic Equipartition Theorem.  

\subsection{Gödel's incompleteness theorem}

\textit{We reproduce the proof of Gödel's first incompleteness theorem in [5].}

\textbf{Definition of consistency and soundness:}

We say that a formal system(definitions, axioms, rules of inference) is \textit{consistent} if no statement which can be expressed in the system can be proved to be both true and false in the system. A formal system is \textit{sound} if only true statements can be proved to be true in the system. Hence, a sound formal system is consistent. 

\textbf{Proof:}

Let $x$ be a finite binary string of length $n$. We say that $x$ is c-random if $K_U(x)> n-c$ for some $c \in \mathbb{N}$. We recall from Lemma 2 that the fraction of sequences that may be compressed by more than $c$ bits is bounded by $2^{-c}$. 

Now, let's consider a sound formal system $F$ that is powerful enough to express the statement $x$ is c-random. Let's suppose $F$ may be described in $f$ bits. By this we mean that there is a fixed-size program of length $f$ such that, when input the number $i$, outputs a list of all valid proofs in $F$ of length $i$. We claim that, for all but finitely many random strings $x$ and $c \geq 1$, the sentence '$x$ is c-random' is not provable in $F$. 

Let's suppose the contrary. Given $F$, we may exhaustively search for a proof that a string of length $n \gg f$ is random, and print it when we find such a string $x$. This procedure, to print $x$ of length $n$ uses only $\log_2 n + f + \mathcal{O}(1)$ bits of data which is much less than $n$. However, $x$ is random due to the proof and the fact that $F$ is sound. Hence, $F$ is not consistent, which is a contradiction. 

\bibliographystyle{plain} 
\bibliography{refs}

\end{document}